\documentclass[11pt]{article}
\usepackage[margin=1in]{geometry}

\usepackage{verbatim}
\usepackage{tikz}
\usetikzlibrary{shapes.geometric}
\usetikzlibrary{quotes}
\usetikzlibrary{arrows.meta}
\usepackage{subfig}

\usepackage{graphicx} 
\usepackage[english]{babel}
\usepackage{amsmath,amsthm,mathtools}

\usepackage{amssymb, amsmath, enumerate, paralist}

\usepackage{graphicx,picins}
\usepackage{algorithmicx}
\usepackage[noend]{algpseudocode}
\usepackage{amsmath,amsfonts,mathrsfs}
\usepackage{xspace}
\usepackage{graphicx}
\usepackage{epsfig}
\usepackage{algorithm}
\usepackage{makeidx}  
\usepackage{hyperref}
\usepackage{todonotes}
\usepackage{comment}
 
\newtheorem{theorem}{Theorem} 
\newtheorem{lemma}{Lemma} 
\newtheorem{claim}{Claim}

\newtheorem{definition}{Definition}

\newcommand{\RR}{\mathbb{R}}

\newcommand{\udg}{{\textsc{UDG}}\xspace}
\newcommand{\fl}{{\textsc{Facility Location}}\xspace}
\newcommand{\pfl}{{\textsc{Prize Collecting Facility Location}}\xspace}

\newcommand{\eps}{\epsilon}

\newcommand{\conn}{{\rm conn}}
\newcommand{\open}{{\rm open}}
\newcommand{\cost}{{\rm cost}}

\newcommand{\opt}{\rm opt}

\DeclareMathOperator{\poly}{poly}

\title{A QPTAS for Facility Location on Unit Disk graphs}
\author{Zachary Friggstad \and Mohsen Rezapour \and Mohammad R. Salavatipour \and Hao Sun\\
Department of Computing Science\\
University of Alberta}
\date{}
\begin{document}

\maketitle
\thispagestyle{empty}
\begin{abstract}
We study the classic \textsc{(Uncapacitated) Facility Location} problem on Unit Disk Graphs (UDGs).
For a given point set $P$ in the plane, the unit disk graph UDG(P) on $P$ has vertex set $P$ and an edge between two distinct points $p, q \in P$ if and only if their Euclidean distance $|pq|$ is at most 1. The weight of the edge $pq$ is equal to their distance $|pq|$.
An instance of \fl on UDG(P) consists of a set $C\subseteq P$ of clients and a set $F\subseteq P$ of facilities, each having an opening cost $f_i$. The goal is to pick a subset $F'\subseteq F$ to open while minimizing
$\sum_{i\in F'} f_i + \sum_{v\in C} d(v,F')$, where $d(v,F')$
is the distance of $v$ to nearest facility in $F'$ through UDG(P).

In this paper, we present the first Quasi-Polynomial Time Approximation Schemes (QPTAS) for the problem. While approximation schemes are well-established for facility location problems on sparse geometric graphs (such as planar graphs), there is a lack of such results for dense graphs. Specifically, prior to this study, to the best of our knowledge, there was no approximation scheme for any facility location problem on UDGs in the general setting.

\end{abstract}

\newpage
\setcounter{page}{1}

\section{Introduction}
Unit-disk graphs (UDGs) are a well-studied class of graphs due to their extensive applications in modeling ad-hoc communication and wireless sensor networks; see for example \cite{akuhn2005local,cabello2015shortest,li2005efficient,gao2003well,kaplan2018routing,yan2012compact,van2005approximation}.
UDGs are defined as intersection graphs of a collection of unit-diamater disks in the two-dimensional plane. Specifically, each \udg represents a set of $n$ unit disks as vertices, with each vertex corresponding to one unit disk.
An edge exists between two vertices/points $p, q$ if and only if their Euclidean distance is at most $1$ (equivalently, the unit-diameter balls around $p$ and $q$ intersect) and the weight or length of each such edge is given by their by the Euclidean distance between the corresponding vertices.

Formally, for a given point set $P$ in the plane, the unit disk graph representation of these points, denoted as UDG(P), is a graph $G = (V, E)$ with the vertex set $V$, where each vertex corresponds to 
a point in $P$. The edge set $E$ consists of edges between points $p$ and $q$ if and only if their Euclidean distance, denoted as $|pq|$, is at most 1. The weight of the edge $pq \in E$ is equal to their distance $|pq|$.
For a given subset $S \subseteq V$, we define the (weak) diameter of $S$ as $\textbf{diam}(S) = \max_{x,y} d_G(x,y)$, where $d_G(x,y)$ represents the minimum weight of a path between vertices $x$ and $y$ in $G$ (we assume $G$ is connected as we may solve facility location for each connected component).

For many optimization problems, approximation schemes are known when the input graph is a \udg (e.g. maximum independent set, minimum dominating set, minimum clique-partition \cite{Matsui,Hunt,Cheng,Imran}). Some of the techniques for designing PTAS's for optimization problems on UDGs (e.g. maximum independent set) involves partitioning the input into regions of bounded size (at a small loss, e.g. ignoring points that touch the boundary of the partitions)
and then solving the problem on such instances using exhaustive search
and/or dynamic programming which leverage Euclidean distance properties. This shifting strategy was introduced by \cite{HM85}.


We study the \fl problem on UDGs. An instance $I=(G,C,F)$ of \fl consists of an edge-weighted graph $G$, where the edges satisfy the metric property, a set $C\subseteq V$ of clients, and a set $F\subseteq V$ of facilities, each having an opening cost $f_i \in \RR^+ $. The goal is to pick a subset $F'\subseteq F$ to open to minimize
$\sum_{i\in F'} f_i + \sum_{v\in C} d(v,F')$, where $d(v,F')$
is the distance of $v$ to nearest facility in $F'$.
\fl has been studied extensively and the best known upper and lower bounds for it are 1.488 \cite{li20131} and 1.46 \cite{guha1999greedy}, respectively. Approximation schemes are known for \fl when the metric is Euclidean \cite{arora1998approximation} or when $G$ is a planar graph \cite{cohen2019polynomial}. Additionally, Cohen et al. \cite{cohen2019local} developed approximation schemes for the ``uniform'' (that is all facilities cost 1 to open) facility location problem in minor-free graphs. To the best of our knowledge, no approximation scheme was known for any facility location type problem on UDGs.


Our main result of this paper is the following:

\begin{theorem}\label{thm:ptas_udg}
  There is an algorithm that, given an instance of \fl in \udg and $\epsilon>0$, finds a $(1+\epsilon)$-approximate solution in time $n^{O_{\epsilon}(\log n)}$, where the constant in $O_\eps(.)$ is $\eps^{-O(\eps^{-2})}$.
\end{theorem}

In order to prove Theorem \ref{thm:ptas_udg} we combine ideas from \cite{cohen2019polynomial} with a low-diameter decomposition
for UDGs that follows from \cite{Lee17,KPRLee} plus a new dissection procedure that is obtained by finding a proper balanced separator for UDGs. This allows us (at a small loss) to break
the problem into independent instances and use dynamic programming to combine the solutions to obtain the solution for the original instance. There are many details on how to put these pieces together carefully so as to bound the overall error.

\section{Preliminaries}
For planar graphs and, more generally, graphs that exclude $K_{r,r}$ as a minor for some fixed $r$, Klein-Plotkin-Rao \cite{klein1993excluded} showed a decomposition of the input graph into low diameter parts by removing a small fraction of edges. More specifically, given a graph $G$ with $n$ nodes and $m$ edges that excludes $K_{r,r}$ as a minor, one can remove $O(mr/\delta)$ edges so that the
(weak) diameter of each remaining component is at most $O(r^2\delta)$. The general idea was based on chopping breadth-first search (BFS) trees (i.e. shortest-path trees in the unweighted version of the graph): suppose one constructs a BFS tree from some root node and then cut the edges at level $i\cdot\delta+r$ for $i\geq 1$ where $r\leq\delta$ is a random offset. Then repeat this procedure on each of the connected components, for $O(r)$ iterations. Then the resulting components have $O(r^2\delta)$ weak diameter. This result was further improved in \cite{fakcharoenphol2003improved,abraham2014cops}, to show for each graph without $K_r$ as a minor there is a probabilistic decomposition into $O(r\delta)$ (weak) diameter components by removing $O(mr/\delta)$ edges.
Lee \cite{Lee17,KPRLee} generalized this by introducing region intersection graphs, which includes UDGs as a special case, and showed that one can obtain similar decomposition results for such graphs. 
Theorem 4.2 in \cite{Lee17} implies that a similar BFS chopping procedure applied to UDGs for a constant number of iterations results in graphs of bounded (weak) diameter.  We describe this chopping procedure a bit more formally.

\begin{definition}[$\delta$-chopping operation]\label{tau-chopping-operation}
For any connected graph $G$ and any number $\delta \geq 1$, we define the $\delta$-chopping operation on $G$ as follows.
Choose any node $x_0$ from $V(G)$, select a random integer $0 \leq r_0 \leq \delta$, and then compute a BFS tree from $x_0$.
Partition $V(G)$ into annuli $A_0, A_1, A_2, \ldots$, where $A_0 = \{ v \in V(G): d'(x_0, v) < r_0\}$ and annulus $A_j$ for $j \geq 1$ is defined as:
$A_j = \{ v \in V(G): r_0 + (j-1)\delta \leq d'(x_0, v) < r_0 + \delta j \}$, where $d'(x_0,v)$ is the number of edges on the BFS tree path from $x_0$ to $v$.

\end{definition}
So there is an offset $r_0$ that cuts only a $1/\delta$-fraction of edges.
Earlier works on minor free graphs \cite{KPRorig} imply that if $G$ is $k_r$-minor free then $O(r)$ iterations of this chopping procedure yields components with (weak) diameter
at most $O(\poly(r)\delta)$. 
The following follows as a corollary of Theorem 4.2
of \cite{Lee17}.

\begin{theorem}[\cite{Lee17}]\label{ActualKPR} 
    $O(1)$ iterations of the $\delta$-chopping iteration applied to a UDG results in a graph of weak diameter $O(\delta)$.  
\end{theorem}

To prove Theorem \ref{thm:ptas_udg} we also rely on the following result for the special case when the instance is in a bounded size region. This will be proved later in Section \ref{sec:ptas} using different techniques than we just discussed.

\begin{theorem}\label{thm:simple_ptas}
There exists a PTAS for \fl in \udg when the point set $P$ is contained within a bounding box of constant size $L=O(1)$ in the plane.
\end{theorem}

Our algorithm for Theorem \ref{thm:ptas_udg} starts with some preliminary steps of algorithm of 
\cite{cohen2019polynomial} that presented a PTAS for \fl on planar metrics. Those preliminary steps in fact are valid for general metrics (they do not use planarity in those initial steps) and reduce the problem to instances with certain structures that we will start from.
For this reason, we briefly outline the main steps of their algorithm.
Their initial steps reduce the problem to instances with certain structural properties and their proof works for general metrics
(not just planar ones). Hence, the same initial reductions work in our
setting as well.

\subsection{Starting point: the PTAS for \fl on planar graphs \cite{cohen2019polynomial}}\label{planar_fl}
Given an instance $I=(G,C,F,f_i)$ of \fl, which consists of  an edge-weighted graph $G$, a set of clients $C$, and a set of facilities $F$ with opening costs $f_i$ (for each $i \in F$),
the first step of their algorithm involves partitioning the instance into separate (independent) sub-instances with specific structural properties. For any solution $D\subseteq F$, we denote by $\conn(D)$
the connection cost of $D$ ($\sum_{c\in C} dist(c,D)$) and by $\open(D)$
the opening cost of facilities in $D$ ($\sum_{i\in D}f_i$), and $\cost(D)=\conn(D)+\open(D)$. We sometimes use $\cost_I(D)$
to denote we refer to the cost of $D$ for instance $I$.
To achieve this, they compute an $\alpha$-approximation solution $\tilde{D}$ (where $\alpha=O(1)$) to a modified instance $\tilde{I}=(G,C,F,\epsilon f_i)$ where each opening cost is scaled down by a factor of $\epsilon$. 
In other words, $\tilde{D}$ is an $O(1)$-approximation for $\tilde{I}$. It is not hard to see that $\tilde{D}$ is also
an $O(1/\epsilon)$-approximation for $I$.
For any $i \in \tilde{D}$, let $cluster(i)$ denote the set of clients connected to $i$ in this solution and define 
$avgcost(i)= \frac{f_i + \sum_{j \in cluster(i)}d_{G}(j,i)}{|cluster(i)|}$ be the average cost of the cluster by facility $i$. 
Suppose $D^*$ is the set of facilities in an optimum solution to $I$. So cost of $(\tilde{D})$ is at most $\frac{1}{\eps}$ times cost of $D^*$. They show:
\begin{lemma}[Corollary 5 \cite{cohen2019polynomial}]
$\forall f\in \tilde{D},\exists g\in D^*: dist(f,g)\leq 2\cdot avgcost(f)$.
\end{lemma}
Then they build a modified instance $I'=(G,C',F,f_i)$ where clients of each $cluster(f)$ are not too close or too far away from $f$ compared to $avgcost(f)$. Let $\opt'$ be the cost of an optimum solution to $I'$ and $\opt$ be the cost of an optimum solution to $I$. They show that:

\begin{lemma}[Corollary 7 \cite{cohen2019polynomial}]
For any $R\subseteq F$ if $\cost_{I'}(R)\leq(1+\gamma)\opt'+\delta$
(for some $\gamma,\delta>0$) then $\cost_{I}(R)\leq(1+2\gamma+8\alpha\epsilon)\opt+\delta$
\end{lemma}

Therefore, a near optimum solution to $I'$ yields a near optimum solution to $I$. In order to be able to obtain a PTAS they further partition the instance (starting from $I$)  into several
instances $I_j$ each of which has certain structural properties.

\begin{definition}[Structured Instance with Bounded Aspect Ratio] \label{structured_instance}
Consider an instance of \fl consisting of  an edge-weighted graph $G=(V,E)$, a set of clients $C\subseteq V$, and a set of facilities $F \subseteq V$ with opening costs $f_i$ (for each $i \in F$).
Suppose we are provided a set $\tilde{D}\subseteq F$ that partitions $C$ into nonempty clusters $\left\{cluster(i) \right\}_{i \in \tilde{D}}$.
We say that the instance has bounded aspect ratio of the average costs and being structured if the following properties hold:
\begin{enumerate}
\item[i)] $\epsilon^{2}\cdot avgcost(i) \leq d_{G}(j,i) \leq \epsilon^{-2}\cdot avgcost(i)$, for each $i \in \tilde{D}$ and each $j \in cluster(i)$,
\item[ii)] for each $i \in \tilde{D}$, there exists $i^* \in D^{*}$ such that $d_{G}(i,i^*) \leq 2\cdot  avgcost(i)$,
\item[iii)] the aspect ratio of the average costs (i.e. 
$\max_{i,j\in \tilde{D}}\frac{avgcost(i)}{avgcost(j)}$) is bounded by $r=\epsilon^{-O(\epsilon^{-2})}$.
\end{enumerate}
\end{definition}

They show that one can partition $I$ into instances $I_j=(G,C_j,F_j,f^j_i)$ such that each instance satisfies the
structural properties of Definition \ref{structured_instance} 
and the following hold:

\begin{lemma}[Lemmas 10 and 11 \cite{cohen2019polynomial}]
Given $D_j\subseteq F$ for $I_j$, we can build $D\subseteq F$ in polynomial time s.t. $\cost_{I'}(D)\leq \sum_j \cost_{I_j}(D_j)+10\alpha\epsilon\opt.$
Furthermore $\sum_j \opt(I_j)\leq(1+9\alpha\epsilon)\opt.$
\end{lemma}

Recall that all these results only use the metric property of instance $I$. It can be verified that these results also hold in UDG setting as well. These together show that to prove Theorem \ref{thm:ptas_udg} it is sufficient to present a QPTAS for instances satisfying
conditions of Definition \ref{structured_instance} and this is what we will do.

%

%

In order to get a PTAS for such instances in the planar case, \cite{cohen2019polynomial} uses a 
Baker-type type layering technique \cite{Baker} in conjunction with the properties of the instance, and further decompose the instance into instances of constant radius at a small  loss. By utilizing balanced separators for planar graphs, they obtain a hierarchical decomposition of the plane embedding of the graph into separate regions (similar to the decomposition of Euclidean instances by Arora \cite{arora1998approximation}). By placing $O(\log n)$ ``portals" along each separator they use
 dynamic programming over this decomposition, while the portals control the interface of different regions.

In our setting, instead of using Baker layering to obtain low diameter instances, we use Theorem \ref{ActualKPR} 
to break the instance that satisfies conditions of Definition \ref{structured_instance} (at a small loss) into low diameter instances. We will use a variant of balanced (partly) separator theorem developed by \cite{YAN2012} for UDGs. Roughly speaking, they show that given a UDG, one can find two paths
originating from a vertex $s$ to two other vertices $x,y$ that are  shortest paths, $P_{s\sim x}$ and $P_{s\sim y}$, such that the removal of these two paths and all the vertices that are within distance
3 of them leaves connected components of size at most $\frac{2}{3}|V(G)|$
(see Theorem \ref{YanLemma}). We use this theorem as an (almost) balanced
separator to obtain a hierarchical decomposition of a low diameter instance into smaller instances. We also place $O(\log n)$ portals 
at these separators and use Dynamic Programming to combine the solutions.
After a logarithmic depth of hierarchical decomposition, we arrive at instances that are easy to solve using other methods (e.g. exhaustive search or another PTAS that is described in Theorem \ref{thm:simple_ptas}).

As mentioned earlier, to prove Theorem \ref{thm:ptas_udg} it is sufficient to give a QPTAS for instances satisfying the conditions of Definition \ref{structured_instance}.

\subsection{Balanced (partly) Separator for UDGs.}
To obtain our dynamic program scheme, we would like to be able
to break a UDG into (almost) balanced parts by picking a separator. This will act as a ``cut" in the disection schema developed by Arora \cite{arora1998approximation} that has been used in designing PTAS's for various optimization problems on Euclidean plane. For this, we utilize the balanced separator theorem for UDGs as presented by Yan et al. \cite{YAN2012}.  
Let $N^{i}_{G}[v] = \{u \in V(G):  d_{G}(v,u) \leq i\}$ and $N^{i}_{G}[S] = \cup_{v \in S}N^{i}_{G}[v]$. A hop-shortest path between two nodes $x,y$ is a path with the minimum number of edges.
\begin{theorem}[\cite{YAN2012}]\label{YanLemma}
    For a \udg $G$, $X \subset V(G)$ and a root $s \in V(G)$, there exist two nodes $x,y$ of $V(G)$ and hop-shortest paths $P_{s\sim  x} = (s,\dots, x)$ and $P_{s\sim  y} = (s, \dots, y)$ for which the removal of $ N^3_{G}[P_{s\sim x}] \cup  N^3_{G}[P_{s\sim  y}] $ from $G$ yields components each having at most $\frac{2}{3}|X|$ vertices from $X$.
\end{theorem}

This theorem serves as a counterpart to the well-known planar balanced shortest paths separator theorem by Lipton and Tarjan. 
However, it poses a challenge due to the fact that the separator is formed by the 3-neighborhoods of the paths. This means that we must not only remove the shortest paths but also all nodes within a distance of three from the nodes on the shortest paths. To tackle this challenge, we narrow our focus to cases where the average distance between clients and facilities is relatively large. By doing so, we can assume that clients and facilities, which end up on two different sides of the shortest path, always get connected via nodes in $V(P_{s\sim  x} \cup P_{s\sim  y})$. 
Note that, as stated in Theorem \ref{TarjanForUDG} (which is a slight modification of this theorem), the only exceptions to this assumption occur when  the path between clients and facilities crosses the border using an edge whose endpoints are very close to nodes in $V(P_{s\sim  x} \cup P_{s\sim  y})$. However, using the fact that we can assume the average distance between clients and facilities is relatively large, we can force those paths to visit nodes in $V(P_{s\sim  x} \cup P_{s\sim  y})$ with a relatively tiny error.

In the following, we demonstrate that a slight modification of this theorem yields a balanced, yet partial in some sense, separator for UDGs. More precisely, we show the following (this follows from Theorem \ref{YanLemma} but we provide a proof for the sake of completeness).

\begin{theorem}\label{TarjanForUDG}
For a \udg $G$, $X \subset V(G)$ and a source $s \in V(G)$, there exists two nodes $x,y$ of $V(G)$ and hop-shortest paths $P_{s\sim x} = (s,\dots, x)$ and $P_{s\sim  y} = (s, \dots, y)$ such that removing  $V(P_{s\sim  x} \cup P_{s\sim y})$ partitions the vertices $V(G \backslash (P_{s \sim x} \cup P_{s\sim  y}))$ into two sets $G_1, G_2$ 
each having at most $ \frac{2}{3} |X|$ vertices from $X$.
Additionally, for any edge $ab \in (V(G_1) \times V(G_2))\cap E(G)$,  there exists $ c \in V(P_{s\sim  x} \cup P_{s\sim  y}) $ such that $d_{G}(a,c), d_{G}(b,c) \leq 4 $.    
\end{theorem}
Note that we say $V(P_{s\sim  x} \cup P_{s\sim  y})$ is a separator between $V(G_1)$ and $(G_2)$ if there are no edges in $G$ that connect a vertex from set $V(G_1)$ to a vertex from set $V(G_2)$.
However, here, we relax this condition and allow for the presence of such edges, provided that their endpoints are in close vicinity to the separator. See Appendix \ref{appendix} for the proof.\\

\section{Proof of Theorem \ref{thm:ptas_udg}}
In this section, we prove Theorem \ref{thm:ptas_udg}.
Consider an instance of \fl, consisting of an edge-weighted \udg $G=(V,E)$, a set of clients $C\subseteq V$, and a set of facilities $F \subseteq V$ with opening costs $f_i$ (for each $i \in F$).
Let $D^*$ denote the set of facilities opened by the optimal solution, with $\opt$ representing the cost of this solution. Additionally, let $\tilde{D}$ denote the $O(1/\epsilon)$-approximation solution as described in Section \ref{planar_fl}, and let $\cost(\tilde{D})$ represent the cost of this solution. Note $\cost(\tilde{D}) = \sum_{i \in \tilde{D}} \Big(f_i + \sum_{j \in cluster(i)}d_{G}(j,i) \Big) = O(\opt/\epsilon)$. 
We assume that the instance satisfies the properties mentioned in Definition \ref{structured_instance}.
Let $r = \epsilon^{-O(\epsilon^{-2})}$ and $N > 0$ denote the minimum distance between a client and its facility (cluster center) in $\tilde{D}$. It can be verified using the properties of Definition \ref{structured_instance} that the inequalities $N \leq d_G(j, f) \leq rN$ and $avgcost(i) \leq rN$ hold for each $i \in \tilde{D}$ and each $j \in cluster(i)$ (these are essentially the conditions in Lemma 9 of \cite{cohen2019polynomial} except that we can't do scaling in UDG instances, hence we have the factor $N$). Moreover, based on property (ii) (Definition \ref{structured_instance}) of the instance, it follows that $d_{G}(j, D^*) \leq 3Nr$ holds for each client $j \in C$. We will use the following observation:
If $T$ is a BFS tree from a vertex $s$ in a UDG $G$, then
for any two vertices $u,v$ at levels $i,i+2$ of the tree respectively (for
any $i$) we have their (weighted) distance in $G$ is strictly larger than $1$ (or else
they would be adjacent and hence cannot be at two levels $i,i+2$ of the BFS tree)
and no more than 2 (as any two adjacent vertices have distance at most $1$), i.e.
$1< d_G(u,v)\leq 2$. Thus, BFS will 2-approximate actual (weighted) shortest paths.

\begin{lemma}
At a loss of $O(\epsilon\cdot \opt)$ we can decompose the instance into a number of independent instances where each has diameter at most $rN/\epsilon^2$.
\end{lemma}
\begin{proof}
Suppose we run a BFS from an arbitrary node $s$ and group the
vertices into layers where layer $i$ consists of all the
nodes of BFS at levels between $(i-1)14N r,\ldots, 14iNr-1$. Note that the ``thickness" of a layer is $14N r$
levels of BFS, so for any two vertices $u,v$ in layers $i,i+2$:
$d_G(u,v)\geq 7Nr$. Since the distance of any client to their facilities in the optimum is at most $3Nr$,
no path from a client to a facility (in the optimum) would cross an entire layer. Now we group consecutive layers into bundles of $\lceil\frac{1}{\epsilon^2}\rceil$ layers, with a random
off-set chosen from $0,\ldots \lceil\frac{1}{\epsilon^2}\rceil$. 
Suppose we call the first layer of each bundle a ``red" layer
and all the other layers of a bundle are blue; so between every
two ``red" layers we have $\lceil\frac{1}{\epsilon^2}\rceil-1$ blue layers.
We open all the facilities of $\tilde{D}$ in the red layers and serve the clients in their cluster. Based on the random shift and the fact that $\tilde{D}$ was a $O(1/\eps)$-approximate solution, the total cost incurred to open these facilities and serve their clients is at most $O(\epsilon^2\cdot\cost(\tilde{D}))=O(\epsilon\cdot\opt)$. So we can assume all these facilities are open (i.e. have zero opening cost) and we delete the clients they have served.
For any other client left in the red layer, they can be partitioned
into two parts: those in the top $7Nr$ levels are called {\em top} red clients and those in the bottom
$7Nr$ levels are called {\em bottom} clients. Since for each client $j$, $d_G(j,D^*)\leq 3Nr$,
the top clients cannot cross over the bottom $7Nr$ levels to be served by a facility. Similarly, the bottom clients cannot be crossing the top $7Nr$ levels to be served by a facility. 
We show how this breaks the instance into independent instances.

For the remaining (blue) layers in every bundle we consider the clients and facilities in those layers, together with the  facilities and remaining clients in the nearest $7Nr$ levels of the two red layers above and below them. Note that these instances are now independent
since for every client $j$, $d_G(j,D^*)\leq 3Nr$, so no client in a blue layer would need to pass beyond $7Nr$ levels into a red layer to reach its facility in optimum. Similarly, the remaining red clients will only need the facilities in the blue layers they are grouped with. This means we can solve the blue layers of each bundle (together with the facilities and clients in a strip of $7Nr$ layers above and below) independently. So we consider the blue layers of each bundle plus the $7Nr$ (red) layers above and below as one instance (recall the facilities in the red layers are open).
This means we can assume we have deleted any connections (edges) between these independent instances.
This is similar to one round of chopping in the proof of Theorem \ref{ActualKPR}.
We perform a sequence of four chopping rounds as above on the graph and utilizing Theorem \ref{ActualKPR}, we can assume that the weak diameter of each independent instance generated is bounded by ${rN}/{\epsilon^2}$ and the total cost paid for the facilities
in the layers chopped is $O(\epsilon\cdot\opt)$; those facilities now have opening cost zero.

So from now on (at a loss of $O(\eps\cdot\opt)$) we focus on each independent instance where the weak diameter is bounded by $rN/\eps^2$. Let's call these instances $H_1, H_2,\ldots$. For each such instance $H_\ell$ we use $C_\ell$ to denote the clients that belong to $H_\ell$. It is easy to see that the $C_\ell$'s are disjoint.
Next, we modify $H_\ell$'s so that they have bounded diameter (not just weak diameter): we add all the vertices of $G$ that are within distance $rN/\eps^2$ of some vertex of $H_\ell$ to $H_\ell$.
Now for each $H_\ell$ we have that the diameter (not just weak diameter) of $H_\ell$ is bounded by $3rN/\eps^2$. Note that the set of clients and facilities of $H_\ell$ is the same as before (we do not  bring in the clients and facilities that were outside of $H_\ell$ when adding vertices to bound the diameter).
\end{proof}

It can also be seen that the sum of optimum solutions of all these instances $H_\ell$'s costs at most $\opt$ since for each client in $H_\ell$ their optimum facility is also in $H_\ell$. We will prove the following lemma:

\begin{lemma}\label{lem:main}
    Given instance $I$ for \fl on a UDG $G$
    let $\tilde{D}$ be an approximate solution as described above and suppose we build instances
    $H_\ell$. There is a quasi-polynomial algorithm
    that produces solutions for $H_\ell$'s such that
    the total cost of the solutions is at most
    $O(\eps^2\cdot\cost(\tilde{D}))+(1+O(\eps))\sum_\ell \opt(H_\ell)$.
\end{lemma}

Now having lemma we immediately get Theorem \ref{thm:ptas_udg} since the total error is at most $O(\eps\cdot\opt)$ since the first term in the expression above is at most $O(\eps\cdot \opt)$ and $\sum_\ell \opt(H_\ell)\leq\opt$.

We treat each individual instance $H_\ell$ separately
and will produce a solution for it of cost $(1+O(\eps)\opt(H_\ell)+E_\ell$ such that $\sum_\ell E_\ell\leq O(\eps^2\cdot\cost(\tilde{D}))$.
A key observation that we use to bound the sum of the additive error bound $E_\ell$ above is the following.
Suppose that $N$ is sufficiently large, i.e. $N>1/\eps^2$
(we will handle the case of small $N$ separately).
Note that since $N |C| \leq \cost(\tilde{D}) = O(\opt/\eps)$, if $N$ is large ($N > 1/\eps^2$) and each
client $c \in C$ is moved an extra $O(1)$ then it adds at most $O(|C|) = O(\eps\cdot\opt)$ to the cost of an
optimum solution for the modified instance; hence this instance still has a solution of cost at most
$(1 + O(\eps))\opt$. 
In our algorithm for each $H_\ell$ we might consider paying an extra $O(1)$ for each client.
Using this argument it can be seen that the total additive error for each $H_\ell$ will be $O(|C_\ell|)$; summing over all the instances $H_\ell$ the additive error is at most $O(\eps\cdot\opt)$.

So from now on we assume our instance is one like a graph $H_\ell$ and we prove Lemma \ref{lem:main}.

\subsection{Hierarchical Decomposition  with Portalization}\label{sec:HD}
In this section, we assume the instance is a low diameter instance
as obtained by applying our chopping operation described above, i.e. one $H_\ell$. More specifically, at a loss of $O(\eps\cdot\opt)$ we assume: $N \leq d_G(j, f) \leq rN$ and $avgcost(i) \leq rN$ hold for each $i \in \tilde{D}$ and each $j \in cluster(i)$, and $d_{G}(j, D^*) \leq 3Nr$ holds for each client $j \in C$, and the diameter is bounded by $3rN/\eps^2$. 

We obtain a hierarchical decomposition of the instance and run a Dynamic Programming based on this decomposition. This decomposition has a logarithmic depth, and each region within the hierarchy is obtained by applying two shortest paths separators to the parent region (for readers familiar with the standard decomposition obtained by a dissection of the plane used in designing PTASs for problems such as TSP and \fl on Euclidean planes: one can think of our shortest paths separators as the ``line" that breaks the problem into two balanced instances; we will place ``portals" at appropriate locations along the separator).
One difference in our schema is that the region at each level might have a ``boundary" that is composed of separators of all the ancestor of it; so it is bounded by $O(\log N)$ separators (whereas, for e.g., a region defined in the recursive decomposition for TSP is defined by 4 dissectig lines). 
This is the only reason our running time becomes quasi-polynomial.\footnote{We conjecture a more careful analysis of our scheme and applying the separator could imply a ``boundary" that is defined by only a few separators. This would turn the whole algorithm into a PTAS.}

Recall $r = \epsilon^{-O(\epsilon^{-2})}$. 
Observe that if we have a UDG with diameter at most $\Delta$ then all the points  must be in a bounding box of $2\Delta\times 2\Delta$. 
Using this, Theorem \ref{thm:simple_ptas} provides a PTAS when the value of $N$ is small (i.e. at most $1/\eps^2$). Therefore, we can assume that $N\geq 1/\eps^2$. 
We provide a solution that has multiplicative approximation factor $(1+O(\eps))$ and additive factor $E_\ell$ that can be charged to the number of clients in the instance (eventually we will have $\sum_\ell E_\ell\leq O(\eps^2\cost(\tilde{D}))$.)

Indeed, assuming that $N$ is sufficiently large is vital for our approach, as it enables us to utilize the balanced ''partly'' separator described in Theorem \ref{TarjanForUDG}. 
This separator allows for a hierarchical decomposition of the graph with a logarithmic depth, but it does permit direct interactions between points (within a small distance from the separator path) in the separated regions. By forcing the interaction through the separator nodes, we incur an additive error. However, when the minimum distance between clients and facilities is sufficiently large, this error becomes negligible.

Let $\Gamma = {3rN}/{\epsilon^2} = O_{\epsilon}(N)$ denote the diameter of the graph. Our goal is to obtain a net of size $\rm{poly}(\Gamma)$ so that the number of points we deal with is in terms of $N$ (instead of $n$), while we loose a small factor (compared to $\opt$). 

\begin{lemma}
We can obtain a graph $G'(V',E')$ where $V'\subseteq V(H_\ell)$ with $|V'|=O(\Gamma^4)$ vertices where:
\begin{itemize}
    \item vertices of $V'$ are at least $1/8$ apart.
    \item A solution for $G'$ can be used to obtain a solution for $H_\ell$ by paying an extra cost at most $O(|C_\ell|)$ where $C_\ell$ is the set of clients of $H_\ell$.
\end{itemize}
\end{lemma}
\begin{proof}
We construct a new graph $G'$ by selecting a subset of vertices from $V(H_\ell)$, denoted as $V'$ ($V'$ will be a ``net"): 
start by adding an arbitrary vertex of $V(H_\ell)$ to $V'$ and then iteratively include nodes from $V(H_\ell)$ that have a minimum Euclidean distance of at least $1/8$ from nodes in $V'$ (alternatively we can think of picking a node from $V(H_\ell)$ to be added to $V'$ and then deleting all the nodes within distance 1/8 of it from $V(H_\ell)$ iteratively). Note that after this round is done, $|V'|$ is $O(\Gamma^2)$. 
Furthermore, for each pair of vertices $u,v \in V'$, 
if there is a vertex $u'$ within distance $1/8$ of $u$ in $V(H_\ell)-V'$ and a vertex $v'$ within distance $1/8$ of $v$ in $V(H_\ell)-V'$  where $u'v'\in E(G)$ then
add both $u',v'$ to $V'$. Note that in this case, $uu',vv',u'v'$ all are edges in $G'$.
This augmentation results in the graph $G'$ (over $V'$) with a total number of vertices of $O(\Gamma^4)$. We arbitrarily order the nodes of $G'$ and for each node $v' \in G'$, we define $B(v')$ as the set of nodes in $H_\ell$ that lie inside the Euclidean ball of radius $1/8$ centered at $v'$ but outside the balls of previously processed nodes. We focus on the UDG induced by $G'$.

To see why we can focus on the net $G'$ and how it helps: 
assume for each node $v\in V'$ that has a client and/or facility in $B(v')$ all its clients/facilities move to $v$ before being connected to any other point
(in other words we move all the clients/facilities of $B(v')$ to $v'$). Note that since we assume $N\leq d_G(j,f)$ in our instance (and since $N>1/\eps^2$) for each $u\in B(v')$, $u$ can ``afford" (at a small loss) to pay the extra $1/8$ to go via its ball center (i.e. $v'$) to be connected to wherever it wanted to go to: it was paying at least $N$ in $\tilde{D}$; so having an extra cost of $1/8$ to go via the center of the ball increases the cost by $O(1)$ per client.
So the total error will be at most $O(|C_\ell|)$ (whereas the cost of optimum  for $H_\ell$ was at least $O(N|C_\ell|)$).
Also, it is easy to see that the size of the graph $G'$, is in terms of $N$ now ($O(\Gamma^4)$).
\end{proof}

Note that the error described above when summed over all $H_\ell$'s, is at most $O(|C|)=O(\cost(\tilde{D})/N)=O(\eps\cdot\opt)$ (since $N|C|\leq\cost(\tilde{D})$).
In other words, if we find a near optimum solution with the assumption that clients/facilities go via the center of the ball they belong to, the extra additive cost over all instances $H_1,H_2,\ldots$ is only $O(\eps\cdot\opt)$. 

We obtain a hierarchical decomposition of $G'$ which will be associated with a rooted binary tree $T$ where the root of $T$ is $V'$. We will have a Dynamic Programming to solve \fl using this decomposition. This will be similar to the hierarchical decompositions obtained for PTAS's on the Euclidean instances (e.g. \cite{arora1998approximation}) except that we use our separator theorem \ref{TarjanForUDG} instead of dissecting lines and that the leaf nodes in our decomposition tree in our case do not correspond to trivial instances (typically a leaf node is a box with only one point in it). In our case, each leaf node is an instance with certain properties which enables us to 
solve it in quasi-polynomial time at a small loss.

Let us describe the decomposition for $G'$. Our hierarchical decomposition tree $T$ is associated with a labeling $\psi: V(T) \rightarrow 2^{V(G')}$: we set the label of the root of $T$ to be $V(G')$. 
Each $t \in T$ represents a subgraph $G'[\psi(t)]$ with the property that if $t_1,t_2$ are children of $t$, then $\psi(t)= \psi(t_1) \cup \psi(t_2)$.
We use $\text{bd}(t)$ to denote the boundary vertices of $\psi(t)$ and the rest of the vertices, $\psi(t)-\text{bd}{(t)}$, we call them the {\em core} vertices and are denoted by $X(t)$. The ``boundary" is obtained by  finding a separator using Theorem \ref{TarjanForUDG} and is added to the boundary that is inherited from the parent (details to follow). The boundary of the root node is the empty set (and so all the vertices of $G'$ are core vertices for the root node of $T$).
The overall idea is whenever a current leaf node $t\in T$ has $|X(t)|>1$ we decompose $G'[\psi(t)]$ into two smaller subgraphs and we obtain children $t_1,t_2$ for $t$.
More specifically, starting from when $T$ is a single node $t$ corresponding to $V(G')$, iteratively, for each leaf $t$ of $T$, if $|X(t)|>1$, apply Theorem \ref{TarjanForUDG} over $G'[\psi(t)]$ with $X=X(t)$ to obtain two graphs $G'_1$ and $G'_2$, along with the two shortest paths $P_{s\sim  x}$ and $P_{s\sim  y}$.
We use the same vertex $s$ to find the two shortest paths $P_{s\sim x},P_{s\sim y}$ for each node $t\in T$ and we make sure this vertex $s$ is passed down. Define $\psi(t_i)$ to be $G[V(G'_i) \cup V(P_{s\sim  x}) \cup V(P_{s\sim  y})]$ and
$\text{bd}(t_i)=\text{bd}(t)\cup P_{s\sim x}\cup P_{s\sim y}$. This also defines $X(t_i)$ to be the subset of vertices of $X(t)$ that fall into $V(G'_i)$ and not in the boundary of $t_i$.
Note that our separator $P_{s\sim  x} \cup P_{s\sim  y}$ separates the vertices $\psi(t)$ of our sub-instance into parts each of which has at most $\frac{2}{3} |X(t)|$ many core vertices.

Every time we find a separator $P_{s\sim x}\cup P_{s\sim y}$ to break the graph $\psi(t)$ (as described) we also designate $m=O_\eps(\log N)$ of the vertices of each of these two paths as {\em portals}. More specifically, let $\delta=O(\eps\Gamma/\log N)$
and designate some of the vertices of these paths as portal so that they are $\delta$ apart (note that as mentioned before, the hop-distance and UDG distances are
within factor 2 of each other). Since these paths have hop-distance length at most $\Gamma$, we will have $O_\eps(\log N)$ portals per path. Our intention is that if two points $u\in X(t_1)$ and $v\in X(t_2)$ are to be connected they have to go through portals of the boundary.

Observe that the depth of $T$ is $h=O(\log \Gamma)=O_\eps(\log N)$ (recall that $|V(G')| = O(\Gamma^4)$).
By construction, for each node $t\in T$, the region defined by $\psi(t)$ consists of core nodes in $X(t)$ and the boundary $\text{bd}{(t)}$ consists of (at most) $h$ separators (which are shortest paths starting from $s$) obtained using Theorem \ref{TarjanForUDG}
each of length proportional to the diameter of the graph, namely $\Gamma$. 
For any $t$ let $\Pi_t$ be the set of portals on the boundary of region $t$. Note that each region boundary
is composed of at most $h=O_\eps(\log N)$ paths and each path has $O_\eps(\log N)$ portals, so each $t$ has at most $O_\eps(\log^2 N)$ portals.

By Theorem \ref{TarjanForUDG} and the way we put the portals (since they are $\delta$ apart), it follows that for any two points $u'\in \psi{(t_1)}$ and $v'\in \psi{(t_2)}$ (where $t_1,t_2$ are children of a node $t\in T$), there exists a portal $\pi$ in the separator that breaks $t$ into $t_1,t_2$ for which $d_{G'}(u',\pi) + d_{G'}(v',\pi) \leq d_{G'}(u',v') + \delta + 4$. 
Now if $u\in B(u'),v\in B(v')$ are two vertices of $H_\ell$ (recall that $B(v')$ is the ball of $v'$ that created net node $v'$ in $G'$), the distance between $u,v$ to go via their ball centers and then via the portal is bounded as well:
$d_G(u,u')+d_{G'}(u',\pi)+d_{G'}(v',\pi)+d_G(v,v')\leq d_{G'}(u',v')+\delta+4.25 \leq d_G(u,v)+\delta+5$.

Suppose we require, for each node $t\in T$ with children $t_1,t_2$, the connection between points in $\psi(t_1),\psi(t_2)$ be via portals in the separator that broke $t$ into $t_1,t_2$. As argued above, this detour adds at most $\delta+5$ to the cost for each connection.
Since the depth of the recursion tree $T$ is $h=O(\log \Gamma)$, the accumulated error incurred to each client/facility for communicating via portals and centers of their balls (among all the levels of the decomposition) is at most $O((\delta+5)\log \Gamma)$. 
Hence, there is a total error of at most $O(|C_\ell|\delta\log \Gamma)$ for all clients connections. By a suitable choice of $\eps'$ depending on $\eps$, and setting $\delta = \frac{\epsilon' \Gamma}{\log \Gamma}$, we get $m = \frac{\log \Gamma}{\eps'} = \frac{\log O_{\epsilon}(N)}{\eps'}$, \emph{the total error} for re-routing connections of clients to be via portals (among all levels of decomposition) between regions is bounded by:
\begin{equation}\label{eq:01}
4|C_\ell|(\eps' \Gamma + 5\log \Gamma) = 4|C_\ell|(\eps'O_{\epsilon}(N) + O_\eps(\log N)).
\end{equation}
This will be our additive error $E_\ell$.
Recall that we have $N \cdot |C| \leq O(\opt/\eps)$, 
implying  $|C|  \leq O(\frac{\opt}{\eps N})$.
This, together with (\ref{eq:01}) and the fact that $\sum_\ell |C_\ell|\leq |C|$, imply that the total additive error for all $H_\ell$'s is bounded by $O(\eps\cdot\opt)$ if $\eps'$ is sufficiently  small.

{\bf Note:} The observation above (that if we re-route clients connections to be via portals adds only $O(\eps\cdot\opt)$
to the total cost) will be crucially used in our DP.
In other words, if for every client/facility connection distance, we have a rounding error of $\delta$ in every level of $T$, then the total error across all $H_\ell$'s is bounded by $O(\eps\cdot\opt)$.
In particular, imagine for each leaf node $t\in T$ we have moved all the points (clients/facilities) in the ball of each node in $\psi(t)$ to the nearest portal in $\Pi_t$.
This adds an extra cost of $O(\eps\cdot \opt)$ over
all levels of decomposition for all $H_\ell$'s. This simplifies our instance significantly and allows us to find a near optimum solution using Dynamic Programming.

{\bf Overview of Dynamic Program based on $T$:}
Here is the overview of our DP based on the hierarchical decomposition $T$. Consider an arbitrary node $t\in T$ and portals $\pi_1,\ldots,\pi_{m'}$ (where $m'=O_\eps(\log^2 N))$ on the paths that define $\text{bd}(t)$. For each portal $\pi_i$ we will have two values $in(\pi_i),out(\pi_i)$: $in(\pi_i)$ indicates (approximately) the distance to the nearest facility (to $\pi_i$) that is supposed to be open in the instance defined by $\psi(t)$, and $out(\pi_i)$ indicates (approximately) the distance to the nearest facility (to $\pi_i$) that will be open outside the region defined by $\psi(t)$ (and clients inside the balls of vertices of
$\psi(t)$ can be connected to them via $\pi_i$ by paying an additional distance cost of $out(\pi_i)$).
These distances are ``approximate" since we only keep multiples of $\delta$, since having a precision parameter of $\delta$ (as argued above) will result in a total error of at most $O(\eps\cdot\opt)$.
For any $t\in T$ and any two vectors $\vec{in},\vec{out}$ of dimensions $m'$ (for the portals of $t$) we will have a DP table entry $A[t,\vec{in},\vec{out}]$. This entry is supposed to store the (approximately) cost of an optimum solution to the \fl instance defined by the balls $B(\psi(t))$ (where $B(S)=\cup_{v\in S} B(v)$) subject to the following conditions:
\begin{itemize}
    \item for each portal $\pi_i$ there is an open facility in the solution with distance to $\pi_i$ at most as specified by $in(\pi_i)$,
    \item for each client either it is served by an open facility inside, or is paying connection cost to go to a portal $\pi_i$ and if $n(\pi_i)$ is the number of clients that go to portal $\pi_i$ then they pay $n(\pi_i)\times out(\pi_i)$ to get serviced by a facility outside of $\psi(t)$ at distance $out(\pi_i)$. This will be part of the cost for the solution.
\end{itemize}

We say vector $\vec{in}$ (and similarly $\vec{out}$) is {\em valid} if it satisfies the following condition: for any two portals $\pi,\pi'$, if their distance (rounded to the nearest
multiple of $\delta$) is $z$ then $|in(\pi)-in(\pi')|\leq z+\delta$. This condition clearly must hold as if there is a facility with distance $in(\pi)$ from $\pi$ then the nearest facility to $\pi'$ cannot be further than $\in(\pi)+z+\delta$ away from it. Our DP table is computed only for valid vectors for each node $t\in T$.

Suppose we have computed (approximate) solutions for each problem defined for each leaf node of $T$ and each (valid) vectors $\vec{in},\vec{out}$ for the portals. Our DP will compute the solution for the instance of root of $T$ (i.e. $V(G')$) in a bottom up manner from the leaf nodes to the root.
For each internal node $t$ with children $t_1,t_2$ and vectors $\vec{in},\vec{out}$ (for $t$) and 
$\vec{in}_1,\vec{out}_2,\vec{in}_2,\vec{out}_2$ (for $t_1,t_2$, respectively) it will see if 
the three solutions are ``consistent". We will define this formally in the next section but at a high level this checks whether the solutions for $t_1,t_2$ (given the portal vectors) can be combined so that we get a solution for $\psi(t_1)\cup\psi(t_2)$ where it is consistent with the vectors of portals specified by $t$.

If one keeps the distances (stored in portal vectors) as as multiples of $\delta$, since distances are bounded by $\Gamma=O_{\eps}(N)$, there will be $O_{\eps}(\log N)$ choices for each portal, which leads to a $(\log N)^{O_{\eps}(\log^2 N)}$ table size.\footnote{We can reduce this using a  trick that is also used earlier (e.g. see \cite{arora1998approximation}) to show how to reduce the size of the portal vectors by storing only ``smoothed" vectors. Informally, the observation that helps is that if we keep distances of the nearest facility (inside or outside) for a portal $\pi_i$ as the multiple of $\delta$
then if the value for portal $\pi_i$ is $\sigma$ then the value for portal just before or after $\pi_{i}$ on the separator path that $\pi_i$ belongs to is in $\{\sigma-1,\sigma,\sigma+1\}$. Thus, the total number of vectors we need to consider for a node $t\in T$ is $O_\eps(\log^2 N\times 3^{m'})=2^{O_\eps(\log^2 N)}$, where there are $O_\eps(\log^2 N)$ choices for the first portal values and then for the subsequent portals there are only $3$ choices for each distance.}

How do we solve the base case of the DP? 
In the base case
we have a leaf node $t\in T$, with $|X(t)|\leq 1$ and a boundary which consists of $O(h)$ separator paths, each with $m$ portals (a total of $m'=O_\eps(\log^2 N)$ portals). We consider (the at most) one node of $X(t)$ as a portal as well.
As mentioned above, at a small loss we assume we have moved each point (client/facility) in $B(\psi(t))$ to the nearest portal of $\psi(t)$ at a small loss. We are given vectors $\vec{in},\vec{out}$, and we have to find the (approximately) best solution for the modified instance over $B(\psi(t))$ (each client/facility is moved to the nearest portal) such that:
\begin{itemize}
    \item the open facilities among the portals satisfy the distance requirements given by $\vec{in}$ (i.e. each portal $\pi_i$ has a facility within distance $in(\pi_i)$). 
    \item For each portal $\pi_i$ if there is an open facility at that portal then all the clients of that portal are served there with zero connection cost; for each portal $\pi_i$ without an open facility the clients of that portal are being served either by the nearest facility among other portals with an open facility, or at a facility at distance specified by $\vec{out}(\pi_i)$.
\end{itemize}

Since there are $O_\eps(\log^2 N)$ portals for $t$, finding the best solution with the condition above
can be done in time $2^{O_\eps(\log^2 N)}$
using exhaustive search: consider any subset $\Pi'$ of portals of $\text{bd}(t)$ 
(where each $\pi\in \Pi'$ has a facility moved to it) and open the cheapest facility on that portal. Then check if the distance requirements (for facilities) are satisfied; pick the best solution among all those subsets whose solution is consistent with the vectors $\vec{in},\vec{out}$.
If there is no such solution the cost stored in the DP table entry is set to $\infty$.

\subsection{Dynamic Program}
In this section we describe the details of our dynamic program based on the decomposition $T$. 
As mentioned earlier, each node $t\in T$ corresponds to a set of vertices $\psi(t)\subset G'$ with boundary nodes
$\text{bd}(t)$ compose of at most $h$ separators, each of which is two paths initiating from $s$ (using Theorem \ref{TarjanForUDG}) and each containing $m$ portals that are $\delta$ apart (for a total of $m'=O_\eps(\log^2 N)$ portals). Note that the diameter of the instance was $\Gamma$, so if we round a distance to the nearest multiple of $\delta$, we get an integer of value at most $O(\Gamma/\delta)=O_\eps(\log N)$.
As mentioned in the overview, for each pair of $m'$-dimension vectors $\vec{in},\vec{out}$, where each entry $in(\pi_i),out(\pi_i)$ (for portal $\pi_i$) is an integer at most $O_\eps(\log N)$,
we have an entry in our DP table $A[t,\vec{in},\vec{out}]$.
As mentioned before, we 
only consider (and have table entries for) valid vectors $\vec{in},\vec{out}$. 

The goal of this subproblem is to identify a set of facilities to open in $B(\psi(t))$ and to assign each client in $B(\psi(t))$ to either an open facility or bring to a portal $\pi_i$ such that minimizes the total opening cost plus connection cost such that:
(i) For each portal $\pi_i$, there is an open facility in $B(\psi(t))$ of distance at most $in(\pi_i)$ (rounded to the nearest multiple of $\delta$); and (ii)
    For any portal $\pi_i$ suppose $n(\pi_i)$ is the number of clients in $B(\psi(t))$ that are connected to $\pi_i$ by the solution. The connection cost for these clients is the cost they pay to be connected to $\pi_i$ plus $n(\pi_i)\times out(\pi_i)$.

We first show the existence of a near optimum solution with certain properties such that our DP actually finds such a near optimum solution.
Starting from an optimum solution $D^*_\ell$ on $H_\ell$ we make some changes to it to satisfy the properties we want, while increasing the cost a little only. First for each node $v\in V'$
with $D^*_\ell\cap B(v)\not=\emptyset$ we consider keeping the cheapest open facility in $B(v)$ and closing all the others and re-routing all the clients that were served by other facilities in $B(v)$ to that single open facility via $v$. This adds at most $1/4$ to the distance each client has to travel (via $v$). This increases the cost by at most $O(|C_\ell|)$ over all clients.
We can assume the open facilities are on the centers of the balls now. Let's call this new (near optimum) solution $D'_\ell$.
Next we modify the solution further by the following process. 
We say a solution is $t$-adapted for a node $t\in T$
if (i) for each of its descendant $t'$ the solution is $t'$-adapted, and (ii) for each client $c\in B(\psi(t))$, if $c$ is connected to a facility outside of $B(\psi(t))$ then it is connected via a portal $\pi\in \Pi$, where $\Pi$ is the set of portals of $t$.

We start with $D'_\ell$ and at leaf nodes of $T$ and going up the tree, we make the solution $t$-adapted for each $t\in T$ at small increase in cost.
Let us consider a leaf node $t\in T$ with $X(t)=w_t$ 
(the case when $X(t)=\emptyset$ is even easier)
and boundary $\text{bd}(t)$ corresponds to paths 
with a total of $m'$ portals $\Pi=\pi_1,\ldots,\pi_{m'}$.
We consider $w_t$ as a portal as well and add it to $\Pi$.
Suppose $\Pi'=\pi_{a_1},\pi_{a_2},\ldots,\pi_{a_\sigma}$ is the subset of $\Pi$, where there is a facility in distance at most $\delta$ of $\pi_{a_i}$ in $D'\cap \psi(t)$. For each
such $\pi_{a_i}$, we assume we have kept open the cheapest facility within $\delta$ of it. For any client $c\in B(\psi(t))$ if $c$ was connected to a facility in $B(\psi(t))$ in $D'$ (note that all of those are within distance $\delta$ of some portal in $\Pi'$) we consider routing $c$ to the nearest portal in $\Pi'$ (and then from there to the single cheap facility we kept open). Note that this increases the connection cost for each client by at most $2\delta$. If
$c$ was connected to a facility outside of $B(\psi(t))$ we re-route it first to a nearest portal $\pi$ along its way and then from there to be connected to the facility it was connected to outside of $B(\psi(t))$ (i.e. we make the connection of $c$ to go through a portal). This increases the connection cost for each $c$ by at most $\delta+5$ as argued earlier in the overview and the solution becomes $t$-adapted.

Now suppose $t\in T$ is a non-leaf node with children $t_1,t_2$
and supposed $D'_\ell$ is $t_1$-adapted and $t_2$-adapted. We make it $t$-adapted with small increase in the cost. For any client
$c\in B(\psi(t))$, if $c$ is connected to a facility in $\psi(t)$ we don't need to make any further changes. Otherwise
we make a detour for the connection of $c$ to go through one of the portals $\Pi$. This detour increases the connection cost
of a client by at most $2\delta$.

One last change we do is in the calculation of the cost of the adapted solution to make it {\bf portal vector} adapted: for each node $t\in T$, the $t$-adapted solution also induces vectors $\tilde{in},\tilde{out}$ for portals of $t$ in the following way.
The clients that are served by facilities outside $\psi(t)$
are first going to a portal $\pi$ (of $t$). The distance from that portal to the nearest open facility (outside $\psi(t)$) rounded up to the nearest multiple of $\delta$ is what induces a value for $\tilde{out}(\pi)$ at node $t$. We use this rounded (up) value instead of the actual distance in the calculation of the cost of the $t$-adapted solution. Similarly, for each portal $\pi$, the distance to the nearest open facility in $\psi(t)$ rounded to the nearest multiple of $\delta$ induces a value $\tilde{in}(\pi)$ for node $t$. Let $\tilde{in},\tilde{out}$ correspond to the vectors induced by the $t$-adapted optimum solution. 
We assume the cost a client pays to go out of $\psi(t)$ to be connected to a facility via a portal $\pi$ (from $\pi$) is $\tilde{out}(\pi)$. 
Note that our estimate of the nearest to $\pi$ open facility inside or outside $\psi(t)$, described by $\tilde{in}(\pi),\tilde{out}(\pi)$ has an additive error of at most $\delta$.
Thus, the connection costs for each client at each node $t$ can be larger by $\delta$ again. We call this new cost, portal vector adapted.

It is easy to see that if we make the solution $t_0$-adapted, where $t_0$ is the root of $T$, and consider the portal adapted cost (as described), then the increase in the connection cost for each client over all the nodes of $T$
is at most $O(\delta\log \Gamma)$ (since the height of $T$ is $O(\log \Gamma)$); summed over all the clients this was shown to be $O(\eps'|C_\ell|\Gamma)$. Thus, the best $t_0$-adapted and portal adapted solution has cost 
$\opt(H_\ell)+O(\eps'|C_\ell|\Gamma)$, and for a suitable choice of $\eps'$ the additive error over all
$H_\ell$'s will be add to at most $O(\eps\cdot\opt)$.

Let $\tilde{in},\tilde{out}$ correspond to the vectors induced by the $t_0$-adapted near optimum solution.
By this argument it is enough to find compute the entries $A[t_0,\vec{0},\vec{0}]$ to obtain a $(1+O(\eps))$-approximate solution.

{\bf Base case:}
Let us consider a leaf node $t\in T$ with $X(t)=w_t$ 
(the case when $X(t)=\emptyset$ is even easier)
and boundary $\text{bd}(t)$, which corresponds to paths 
with a total of $m'$ portals $\Pi=\pi_1,\ldots,\pi_{m'}$.
We consider $w_t$ as a portal as well and add it to $\Pi$.
For any subset $\Pi'=\pi_{a_1},\pi_{a_2},\ldots,\pi_{a_\sigma}$ of $\Pi$, where there is a facility in $B(\pi_{a_i})$, we consider opening the cheapest facility in $B(\pi_{a_i})$; let's call that facility $f(a_i)$.
For any client $c\in B(\psi(t))$ we consider routing $c$
to (i) nearest portal with an open facility, or (ii) to a portal $\pi_{a_i}$ to be connected outside at a total cost $d_G(c,\pi_{a_i})+out(\pi_i)$ if this is less than the distance to the nearest portal with an open facility.
This will be considered a feasible solution if:
for each portal $\pi_i$ there is a portal $\pi_{a_j}\in\Pi'$ with an open facility such that $d_G(\pi_i,\pi_{a_j})\leq in(\pi_i)$.
The cost for $A[t,\vec{in},\vec{out}]$ will be the cost of the cheapest feasible solution over all choices of $\Pi'$ as described above. If there is no such solution (consistent with vectors $\vec{in},\vec{out}$, we set $A[t,\vec{in},\vec{out}]=\infty$.
It is easy to see that we obtain the best $t$-adapted portal adapted solution.

{\bf Filling in the rest of DP table:} 
Now consider an arbitrary (non-leaf) node $t\in T$ and vectors $\vec{in},\vec{out}$ and suppose $t$ has children $t_1,t_2$.
Suppose all the entries of $t_1,t_2$ for all vectors of portals are computed. Let $\Pi,\Pi_1,\Pi_2$ be the set of portals of $t,t_1,t_2$, respectively.
For vectors $\vec{in}_1,\vec{out}_1$ for portals of $t_1$, and
vectors $\vec{in}_2,\vec{out}_2$ for portals of $t_2$ we say
subproblems $(t,\vec{in},\vec{out})$, $(t_1,\vec{in}_1,\vec{out}_1)$, $(t_2,\vec{in}_2,\vec{out}_2)$ are consistent if the following hold:

\begin{itemize}
    \item For each portal $\pi\in \Pi$, 
    either $\vec{in}_1(\pi)=\vec{in}(\pi)$ or $\vec{in}_2(\pi)=\vec{in}(\pi)$.
    \item For each portal $\pi\in (\Pi_1\cap\Pi_2)-\Pi$, i.e. a portal that is on the separator of $t$ that creates $t_1,t_2$: $\vec{in}_1(\pi)=\vec{out}_2(\pi)$ and $\vec{in}_2(\pi)=\vec{out}_1(\pi)$.
    \item for each $\pi\in (\Pi_1\cap \Pi_2\cap \Pi)$ we must have
    $\vec{out}_1(\pi)=\vec{out}_2(\pi)=\vec{out}(\pi)$.
\end{itemize}

First observe that checking consistency for the three subproblems can be done in time $\poly(m)$. Then 

\[ A[t,\vec{in},\vec{out}]=\min\{
A[t_1,\vec{in}_1,\vec{out}_1]+A[t_2,\vec{in}_2,\vec{out}_2]\},
\]

where the $\min$ is over all vectors $\vec{in}_1,\vec{out}_1,\vec{in}_2,\vec{out}_2$ such that $(t,\vec{in},\vec{out})$, $(t_1,\vec{in}_1,\vec{out}_1)$, $(t_2,\vec{in}_2,\vec{out}_2)$ are consistent. By induction, assuming that $\vec{in},\vec{out},\vec{in}_1,\vec{out}_1,\vec{in}_2,\vec{out}_2$ are induced portal vectors for a $t$-adapted near optimum solution $D'$ and that $A[t_1,\vec{in}_1,\vec{out}_1]$ and $A[t_2,\vec{in}_2,\vec{out}_2]$ are computed correctly, one can see that we get that the cost of a solution at $A[t,\vec{in},\vec{out}]$ that is no more than that of optimum $t$-adapted solution at induced on $B(\psi(t))$.

{\bf Run time analysis:} Note that diameter of a (connected) UDG is at most $n$ (and we assume the graph is connected since we can run the algorithm  on each connected component). Thus $N=O(n)$ and hence the size of the DP is $2^{O_\eps(\log^2 N)}=2^{O_\eps(\log^2 n)}$. To compute each base case it takes $2^{O_\eps(\log^2 n)}$ time and for each non-leaf node of $t$ which children $t_1,t_2$ and vectors $\vec{in},\vec{out}$, the solution $A[t,\vec{in},\vec{out}]$ can be computed by comparing all triples of valid solutions for $t,t_1,t_2$; this also takes $2^{O_\eps(\log^2 n)}$. Overall the runtime is therefore $n^{O_\eps(\log n)}$, where the constant in $O_\eps(.)$
is $\eps^{-O(\eps^{-2})}$.

\subsection{PTAS for \fl on UDG in Bounded Regions}\label{sec:ptas}
In this section we present a PTAS for \fl in \udg in the special case that the point set $P$ is contained within a bounding box of size $L \times L$ in the plane where $L$ can be regarded as a constant, i.e. prove Theorem \ref{thm:simple_ptas}

Consider an instance of \fl consisting of  an edge-weighted graph $G = UDG(P)$, a set of clients $C \subseteq P$, and a set of facilities $F \subseteq P$ with opening costs $f_i$ (for each $i \in F$). Let $D^* \subseteq F$ be the facilities in an optimum solution and let $i^{*}_{j} \in D^*$ denote the facility that serves the client $j$ in $D^*$. Suppose we know $D^*$
(this assumption will be removed) and let $\epsilon > 0$ be the error parameter. We greedily form an \emph{$\epsilon$-net $F'$} as follows:
\begin{itemize}
\item Sort the facilities in $D^{*}$ by their opening costs in non-decreasing order.
\item In this order, while there is some $i \in D^{*}$ such that $d_G(i, F') > \epsilon$, add $i$ to $F'$.
\end{itemize}
This procedure ensures that the resulting $F'$ forms an {$\epsilon$-net}, where no facility in $D^*$ is at a distance greater than $\epsilon$ from $F'$.
\begin{claim}
$|F'| \leq O(L/\epsilon^2)$.
\end{claim}
\begin{proof}
The balls of radius $\epsilon/2$ centered at each point in $F'$ are interior-disjoint. These balls collectively occupy a total area of $\Omega(\epsilon^2 \cdot |F'|)$. The proof follows from the fact that the balls are entirely contained within a square of side length $L+\epsilon$.
\end{proof}

Now we divide the instance into sub-instances using a random grid of size $1/2$ that splits the bounding box into squares of size at most $1/2 \times 1/2$.  Note $d_G(i,j) \leq 1$ for any two points $i$ and $j$ lying in the same cell. As a result, the metric between points within a cell can be treated as Euclidean distance.
For any client $j \in P$, we say $j$ is {\bf cut} if $j$ and $i^{*}_{j}$ lie in different grid cells. Let $c^*_j = d_{G}(j, i^{*}_{j})$.
\begin{claim}
For any point $j \in P$, ${\bf Pr}[j \text{ is cut}] \leq 2 c^*_j$.
\end{claim}
\begin{proof}
This is obvious if $c^*_j\geq 1/2$, so let's assume
$c^*_j<1/2$.
The probability that a horizontal line in the grid separates points $p,q$ is at most $|pq|$. Same when considering vertical lines in the random grid. Since $c^*_j <1/2$, then the centre serving $j$ has a direct connection with $j$ so $c^*_j = d_G(j,i^*)$ as well.
\end{proof}

For each grid cell $c$, let $X_c$ be the restriction of the points in the input to cell $c$. 
Recall that, in the prize-collecting version of the facility location problem (\pfl), in addition to the input for facility location, each client $j$ is associated with a penalty cost $\pi_j$. This penalty cost can be paid instead of the connection cost. The goal is to find an optimal solution that minimizes the total cost, including opening costs and both connection costs and penalties.
Define an Euclidean \pfl instance for each cell $c$. For $j \in X_c$, its penalty is $\pi_j := d_G(j, F')$.
Let $D_c^* := (D^*-F') \cap X_c$ be the optimum facilities in cell $c$ that are not part of the net.

\begin{claim}
The optimum \pfl solution for this instance has cost at most
\[ \sum_{i \in D_c^*} f_i + \sum_{j \in X_c} c^*_j + \sum_{j \in X_c : j \text{ cut}} \epsilon. \]
\end{claim}
\begin{proof}
Consider the solution that opens $D^*_c$. If a point $j$ is not cut, we can directly connect it to its centre in $D^*_c$ paying a cost of $c^*_j$ (since the cell $c$ has dimensions $1/2 \times 1/2$ then the direct connection is possible).

Otherwise, we can pay the penalty for $j$. Note this is upper bounded by moving $j$ to its optimum centre in $D^*$ (paying $c^*_j$) and then from there to the nearest net point (paying an additional $\epsilon$).
\end{proof}

\begin{proof}[Proof of Theorem \ref{thm:simple_ptas}]
Consider the following algorithm:
\begin{enumerate}
\item For all possible choices of $F' \subset F$ with $|F'| \leq O(L/\epsilon^2)$ do
\begin{itemize}
\item Partition the instance into sub-instances using a random grid of size $1/2$ . Let $\mathcal C$ be the corresponding cells.
\item For each $c \in C$ run a PTAS on the corresponding Euclidean \pfl instance \cite{cohen2021near}.
\item Obtain a solution for the facility location instance: open all facilities in set $F'$, as well as the facilities opened by PTASs in each cell. For every $j \in P$ that paid a penalty in its corresponding \pfl instance, assign $j$ to its nearest (in the UDG metric) facility in $F'$.
\end{itemize}
\item Output a minimum cost solution, among the solutions obtained.
\end{enumerate}

It is sufficient to show that \emph{$\epsilon$-net $F'$} satisfies the claim. Using the previous claims, we obtain the following. The cost of opening all facilities in \emph{$\epsilon$-net $F'$} plus expected total cost of all \fl solutions for all cells $c$ is at most:
\[\sum_{i \in D^*} f_i  + \sum_{j \in X_c} (1+2 \cdot \epsilon) \cdot c^*_j \leq (1+ O(\epsilon)) \big( \sum_{i \in D^*} f_i  + \sum_{j \in P} c^*_j \big)\]
using the fact that for each client $j$, we always pay $c^*_j$ and, perhaps, an additional $\epsilon$ if $j$ is cut. But $j$ is cut with probability at most $2 \cdot c^*_j$.
\end{proof}

{
\bibliography{soda_bibfile}
}
\section{Appendix}\label{appendix}

\begin{proof}[{Proof of Theorem \ref{TarjanForUDG}.}]
Consider shortest paths $P_{s\sim  x} = (r,\dots, x)$ and $P_{s\sim  y} = (r, \dots, y)$ by \autoref{YanLemma}, partition the components of $ G \backslash ( N^3_{G}[P_{s\sim  x}] \cup  N^3_{G}[P_{s\sim  y}]  ) $ into two graphs $G'_1$ and $G'_2$ each containing at most $\frac{2}{3}|V(G)|$ vertices.
Partition $ (N^3_{G}[P_{s\sim  x}] \cup  N^3_{G}[P_{s\sim  y}]) \backslash (P_{s\sim  x} \cup P_{s\sim  y}) $ into two sets $S_1,S_2$ such that each $G'_i \cup S_i$ (for $i = 1, 2$) has size at most $\frac{2}{3}|V(G)|$.

We claim that $G_i=G'_i \cup S_i$ satisfies the required property.
Let  $  ab \in V(G_1) \times V(G_2)  $, then it cannot be that $a \in G_1$ and $b \in G_2$ by \autoref{YanLemma}. Without loss of generality, let  $b \notin G_2$. Then there exists $ c \in V(P_{s\sim  x} \cup P_{s\sim  y}) $  and a path $Q$ of length at most 3 from $b$ to $c$ (see \autoref{YanLemma}). Then $abQc$ is a path of length at most 4 from $a$ to $c$.
 \end{proof}

\end{document}